\documentclass[11pt,a4paper]{article}

\usepackage[T1]{fontenc}
\usepackage[utf8]{inputenc}
\usepackage{lmodern}
\usepackage{microtype}
\usepackage[margin=2.5cm,marginparwidth=2cm]{geometry}

\usepackage{stmaryrd}
\usepackage{amsfonts}
\usepackage{amsmath}
\usepackage{amssymb}
\usepackage{amsthm}
\usepackage{mathtools}
\usepackage{xcolor}
\usepackage{graphicx}
\usepackage{bm}
\usepackage{tikz}
\usepackage{float}
\usepackage{caption}
\usepackage{subfig}
\usepackage{tabularx}
\usepackage[disable]{todonotes}
\usepackage{xspace}           
\usepackage{authblk}          
\usepackage[backend=biber,maxbibnames=10,doi=false,url=false,isbn=false]{biblatex}
\usepackage[inline,shortlabels]{enumitem}
\usepackage[ruled,linesnumbered,vlined]{algorithm2e}
\usepackage{adjustbox}
\usepackage[colorlinks=true,linkcolor=blue!60!black,citecolor=green!40!black,
            urlcolor=blue!60!black]{hyperref}
\usepackage{thm-restate}

\theoremstyle{plain}
\newtheorem{thm}{Theorem}[section]

\newenvironment{introthm}[1]
  {\par\medskip\noindent\textbf{Theorem~#1.}\itshape}
  {\par\medskip}

\newtheorem{lem}[thm]{Lemma}
\newtheorem{prop}[thm]{Proposition}
\theoremstyle{definition}
\newtheorem{defn}[thm]{Definition}
\newtheorem{rem}[thm]{Remark}
\newtheorem{example}[thm]{Example}

\newcommand{\sage}{\textsc{SageMath}\xspace}
\newcommand{\R}[0]{\mathbb{R}}
\newcommand{\F}{\mathbb{F}}

\newcommand{\Z}[0]{\mathbb{Z}}
\newcommand{\Q}[0]{\mathbb{Q}}

\newcommand{\NPV}{N_{P(\F_p^n)}}

\DeclareMathOperator{\card}{card}
\DeclareMathOperator*{\argmin}{arg\,min}
\newcommand{\vp}{v_p}
\newcommand{\gr}{\mathrm{gr}}
\newcommand{\Vect}{\mathrm{Vect}}

\newcommand{\val}{\textnormal{val}}

\title{On the Complexity of Computing with \ensuremath{p}-adic Lattices}

\author[1]{Pierre Fromonot}
\author[2]{Mercedes Haiech}
\author[2]{Tristan Vaccon}
\affil[1]{Universit\'e de Limoges, 87060 Limoges, France}
\affil[2]{Universit\'e de Limoges, CNRS, XLIM UMR 7252, 87060 Limoges, France}
\affil[ ]{\small\texttt{pierre.fromonot@etu.unilim.fr},
  \texttt{mercedes.haiech@unilim.fr}, \texttt{tristan.vaccon@unilim.fr}}

\date{}

\begin{document}

\maketitle

\begin{abstract}
We study the complexity of two fundamental problems attached to
$p$-adic lattices: computing an $N$-orthogonal basis, and deciding
whether a given basis is $N$-orthogonal.
When the norm is supplied through an explicit diagonal
representation, both problems are polynomial. However, in the black-box model, where
the ultrametric norm $N$ is fixed and only accessible through an evaluation oracle, we prove that both problems require $\Omega(p^{\nu-1})$ queries in the worst case, where $\nu \leq n$ is the largest multiplicity of a commensurability class of $N$ and $n$ is the dimension of the ambient vector space.
We then analyse the optimality of these bounds by providing
explicit algorithms.

Our results answer formally an open question of Zhang, Deng and Wang
\cite{zhang2025notes} and refine the recent results of Duan and Wu \cite{duan2026norm}  by providing pointwise query bounds depending on the norm's structure, in addition to uniform bounds.
\end{abstract}

\medskip
\noindent\textbf{Keywords.} $p$-adic lattices, ultrametric norms,
orthogonalisation, query complexity, LVP, CVP, post-quantum cryptography.

\widowpenalty = 10000

\section{Introduction}

\subsection{History and related results}

In 2018, Deng, Luo and Xiao~\cite{deng2018some} introduced the Longest Vector Problem (LVP) and the Closest Vector Problem (CVP) for $\Z_p$-lattices  in finite extensions $K/\Q_p$, equipped with the canonical absolute value extending $|\cdot|_p$,
motivated by the Euclidean lattice problems and their role in post quantum cryptography (e.g.\ Kyber, Dilithium) ~\cite{NIST:PQCrypto}.
In such spaces, the ultrametric norm gives a radically different geometry from that of their Euclidean counterparts and opens the way
to new cryptographic hardness assumptions.

Zhang, Deng and Wang ~\cite{zhang2025norm} extended this framework to finite-dimensional
$\Q_p$-vector spaces equipped with arbitrary ultrametric norms $N$. They proved that
LVP and CVP are solvable in polynomial time as soon as
one is given an $N$-orthogonal basis, a structure whose existence is
always guaranteed in ultrametric geometry~\cite{weil1967basic}, and
that CVP and the computation of such a basis are polynomially
equivalent.

More recently, Zhang et al.~\cite{zhang2025notes} showed that the usual
norm $\vert \cdot \vert_p$ on a field extension $K/\Q_p$ reveals a
uniformiser, from which an orthogonal basis can be built efficiently:
this made LVP and CVP trivially solvable in polynomial time in that setting, and completely broke the
first cryptosystem based on these structures~\cite{deng2024public}.
Recalling that every ultrametric norm admits a ``diagonal'' form, out
of which an $N$-orthogonal basis can be read off immediately, the
authors raised the following fundamental question: \emph{is it possible
to find a diagonal form (or an orthogonal basis) of a given norm when
one only has access to an evaluation oracle?}

This paper gives the answer: one can compute
an orthogonal basis, with at least exponentially many queries
in the worst case.

Independently and very recently, Duan and Wu~\cite{duan2026norm} studied
the same norm-query model. They prove that orthogonalisation admits no
uniform query bound depending only on the dimension: for every
deterministic algorithm that outputs an $N$-orthogonal basis for every
ultrametric norm $N$, the number of queries is unbounded as $N$ varies,
even at fixed dimension.

While Euclidean lattice-based cryptography dominates the
post-quantum landscape, $p$-adic lattices offer a potential alternative
should the hardness assumptions behind LWE turn out to be weakened, as
suggested by the recent work of Daniel R.\
Simon~\cite{cryptoeprint:2026/1591}.

\subsection{Contributions and organisation}

The first observation is that \emph{$p$-adic commensurability}—the equivalence relation $v_1 \sim v_2$ when $N(v_1)/N(v_2) \in p^{\Z}$—is the right parameter governing the difficulty of the problems considered here: orthogonalization and orthogonality decision. In Section~\ref{sec:family}, we construct an explicit family, denoted $\NPV$, of $\frac{p^n-1}{p-1}$ norms whose distinction requires exponentially many oracle queries.

Let $\nu_{\max} \leq n$ denote the largest multiplicity of a commensurability class of $N$. The case $\nu_{\max}=n$, which yields the strongest lower bound below, includes many natural norms like the usual sup norm and every norm in the constructed family $\NPV$. We use $\Omega$ to denote asymptotic lower bounds on the number of oracle calls. Our first main result, proved in Section~\ref{sec:ortho}, is:
\begin{introthm}{\ref{thm:lowerbound-ortho}: Lower Bound Orthogonalization}
    Any algorithm (deterministic, or randomised with success probability
    $\geq 2/3$) which reconstructs an $(A,c)$-representation of $N$ from a
    norm oracle performs $\Omega(p^{\nu_{\max}-1})$ oracle calls in the
    worst case. In particular, when $\nu_{\max} = n$, the algorithm performs
    $\Omega(p^{n-1})$ oracle calls.
\end{introthm}
We provide with Algorithm \ref{algo:descend} an orthogonalisation algorithm
with a query complexity of $O(nD p^{\nu_{\max}})$
where $D$ is defined to be the geometric conditioning of
the input basis.

For the decision problem, the relevant parameter is instead the size of the largest commensurability block in the input family, denoted $l_{\max}$. In particular, for the usual sup norm and every norm in $\NPV$, all nonzero vectors are commensurable. Thus, a family of $m$ vectors has $l_{\max} = m$, and a basis $l_{\max} = n$.
This leads to our second main result (Section~\ref{sec:decide}):
\begin{introthm}{\ref{thm:lowerbound-decide}: Lower Bound Decision}
    Any algorithm (deterministic, or randomised with success probability
    $\geq 2/3$) that decides whether a basis $B$ is $N$-orthogonal performs
    $\Omega(p^{l_{\max}-1})$ oracle calls in the worst case. In particular,
    when $l_{\max} = n$ the algorithm performs $\Omega(p^{n-1})$ oracle calls.
\end{introthm}

We show that these bounds are tight (up to polynomial factors).
In the single-block case $l_{\max} = n$, Algorithm \ref{algo:decide-oracle} makes $O\left(\frac{p^n - 1}{p - 1}\right) = O(p^{n-1})$ oracle calls. Combined with the lower bound, we get a tight bound of $\Theta(p^{n-1})$.

Our exponential lower bounds do not result from an unbounded or degenerate family of norms. The family
$\NPV$ is finite, explicit, and all its members take their values in
$p^{\Z}$. Nevertheless, distinguishing its elements already requires $\Omega(p^{n-1})$ queries. We identify two distinct
parameters that govern the cost of orthogonalization. The multiplicity $\nu_{\max}$ of the largest commensurability class yields an intrinsic barrier
$\Omega(p^{\nu_{\max}-1})$, whereas the
geometric conditioning $D$ of the instance (defined in Section~\ref{sec:descend}) accounts for the failure of a uniform bound over all norms, as shown by Duan \& Wu in~\cite{duan2026norm}.

By contrast, when a diagonal representation $(A,c)$ is given, our \emph{residual commensurability} criterion (Section~\ref{sec:decide}) decides orthogonality in $O(nm^2)$ arithmetic operations: hardness lies in the norm oracle, not only in the norm structure itself.
\begin{introthm}{\ref{thm:residual}: Residual Commensurability Criterion }
      Let $N$ be a norm on $V$ and let $b_1,\dots,b_m$ be nonzero vectors of
    $V$. Let $B_1,\dots,B_k$ be their commensurability classes and
    $\bar M_1,\dots,\bar M_k$ the associated residual matrices of leading
    digits. Then the family is $N$-orthogonal if and only if
    $\ker \bar M_i = \{0\}$ over $\F_p$ for every block $B_i$.
\end{introthm}

\sage implementations of all our algorithms are publicly available at:
\url{https://github.com/Synthettikk/p-adic-Lattices-for-Cryptography}.

\section{Norms and \ensuremath{p}-adic lattices}
\label{sec:prelim}

\subsection{The norm-oracle model}

Throughout the paper we consider a black-box norm oracle: it takes as
input a vector in $\Q_p^n$ and returns the exact value of its norm.
For each fixed norm $N$, we assume that its values are taken in a discrete subset $\Gamma \subset \mathbb{R}$, allowing exact equality tests, order comparisons, tests for membership in $p^{\Z}$, and usual arithmetic operations.

We distinguish oracle-query complexity from arithmetic complexity. For algorithms operating through a norm oracle, our bounds concern the number of oracle queries, while the associated arithmetic work is measured separately and may depend on that number of calls. Each individual arithmetic operation is assumed to have polynomial cost.

\subsection{Basic facts on \ensuremath{p}-adic numbers}

We recall the minimal background on $p$-adic numbers; more details can
be found in~\cite{serre1968corps}. Let $p$ be a prime number. For
$x \in \Q^*$, write $x = p^t \cdot \frac{a}{b}$ with $t,a,b \in \Z$,
$b \neq 0$ and $p$ dividing neither $a$ nor $b$, and set
$\vert x \vert_p = p^{-t}$, $\vert 0 \vert_p = 0$. Then
$\vert \cdot \vert_p$ is a non-archimedean absolute value on $\Q$: it
is multiplicative, it vanishes only at $0$, and it satisfies
$\vert x+y \vert_p \leq \max \{ \vert x \vert_p, \vert y \vert_p \}$,
with equality whenever $\vert x \vert_p \neq \vert y \vert_p$. We call
this last property the \emph{ultrametric inequality} (or the
\emph{ultrametric equality} when it is one).

We denote by $\Q_p$ the completion of $\Q$ for $\vert \cdot \vert_p$,
by $\Z_p = \{ x \in \Q_p : \vert x \vert_p \leq 1 \}$ its ring of
integers --- a compact discrete valuation ring, with maximal ideal
$p\Z_p$ and valuation $\vp$ --- and by
$\Z_p^{\times} = \{ x \in \Q_p : \vert x \vert_p = 1 \}$ its group of
units. For a positive integer $r$, we denote by $P(\F_p^{r})$ the
projective space of $\F_p^{r}$, which we simply regard as the set of
one-dimensional subspaces (the vector lines) of $\F_p^{r}$.
It is well-known that
$\card ( P(\F_p^{r}) ) = 1+p+\cdots+p^{r-1} = \frac{p^{r}-1}{p-1}$.

\subsection{Norms and orthogonal bases}

Since $\Q_p$ is a field, $\Q_p$-vector spaces are defined as usual; we
denote them by $V$, and assume throughout that
$\dim_{\Q_p} V = n \geq 2$ is finite. A \emph{norm} $N$ on $V$ is a map
$V \to \R_{\geq 0}$ such that $N(v) = 0$ if and only if $v = 0$,
$N(\alpha v) = \vert \alpha \vert_p N(v)$ for all $\alpha \in \Q_p$ and
$v \in V$, and $N(v+w) \leq \max\{N(v),N(w)\}$ for all $v,w \in V$,
with equality whenever $N(v) \neq N(w)$ --- the ultrametric inequality
again. Note that, by induction, it yields
$N(v_1 + \dots + v_k) = \max_i N(v_i)$ as soon as the $N(v_i)$ are
pairwise distinct.

\begin{prop}[Weil,~\cite{weil1967basic}]
    \label{prop:weil}
    Let $N$ be a norm on $V$. Then there is a decomposition
    $V = V_1 \oplus V_2 \oplus \cdots \oplus V_n$ of $V$ into a direct sum
    of $n$ one-dimensional subspaces such that
    \begin{equation}
        \label{eq:Northo}
        N\Big( \sum_{i=1}^{n} v_i \Big) = \max_{1 \leq i \leq n} N(v_i)
    \end{equation}
    for all $v_i \in V_i$.
\end{prop}

\begin{defn}[$N$-orthogonal basis,~\cite{weil1967basic}]
    Let $N$ be a norm on $V$. A basis $(b_1,\dots,b_n)$ of $V$ is
    \emph{$N$-orthogonal} if the decomposition
    $V = V_1 \oplus \cdots \oplus V_n$ with $V_i = \Vect_{\Q_p}(b_i)$
    satisfies~\eqref{eq:Northo}. A family $(b_1,\dots,b_m)$ of nonzero
    vectors is $N$-orthogonal if it is an $N$-orthogonal basis of the
    subspace it spans.
\end{defn}

Thus Proposition \ref{prop:weil} states that for any norm $N$, $V$ admits an $N$-orthogonal basis.
The directness of the sum is in fact a consequence of~\eqref{eq:Northo}:
if $\sum \alpha_i b_i = 0$ then
$\alpha_i b_i = 0$ for every $i$. As such, proving~\eqref{eq:Northo} is enough to characterise orthogonality.

\begin{thm}[Goldman--Iwahori,~\cite{zhang2025notes}]
    \label{thm:diagonal}
    Let $(b_1,\dots,b_n)$ be a basis of $V$. Then every norm $N$ on $V$ can
    be written as
    \[ N\Big( \sum_{i=1}^{n} \alpha_i b_i \Big)
    = \max_{1 \leq r \leq n} c_r \Big\vert \sum_{i=1}^{n} \alpha_i a_{r,i}\Big\vert_p, \]
    where $\alpha_i \in \Q_p$, $c_r \in \R_{>0}$
    and $(a_{r,i})_{n \times n} \in \mathrm{GL}_n(\Q_p)$.
\end{thm}
\begin{proof}
By Proposition~\ref{prop:weil} there is an $N$-orthogonal basis
$(e_1,\dots,e_n)$ of $V$; put $c_i := N(e_i)$ and let
$A = (a_{r,i}) \in \mathrm{GL}_n(\Q_p)$ be the change of basis matrix,
$(b_1 \cdots b_n) = (e_1 \cdots e_n) \cdot A$.
\end{proof}

We call \emph{$(A,c)$-representation} of a norm $N$ on $V$ its
expression in the form of Theorem~\ref{thm:diagonal}, and say that $N$
is given in \emph{diagonal form} when it is presented by such a pair
$(A,c)$.

The proof makes it clear that knowing an $(A,c)$-representation of $N$
is the same as knowing an $N$-orthogonal basis: the columns of $A^{-1}$
are the coordinates, in the basis $(b_1,\dots,b_n)$, of the vectors of
an $N$-orthogonal basis.

Since $A$ is invertible, its rows define a basis of the dual space
$V^{\vee}$. We write $L_r$ for the corresponding $r$-th linear form, i.e.\
$L_r(v) = a_{r,1}v_1 + \cdots + a_{r,n}v_n$ for every
$v = \sum_{i=1}^n v_i b_i \in V$. The norm then reads
$N(v) = \max_{1 \leq r \leq n} c_r \vert L_r(v)\vert_p$.

\begin{example}
    \label{ex:norm}
    Consider the norm defined, for every
    $v = (v_1,v_2,v_3)^T \in \Q_p^3$ expressed in a basis
    $(b_1,b_2,b_3)$ of $V$, by
    \[N(v) = \max \big(\pi \vert v_1\vert_p,\ \sqrt{2}\,\vert v_1+v_2 \vert_p,\
      \vert v_1+v_2+v_3\vert_p\big),\]
    that is, $c_1 = \pi$, $c_2 = \sqrt{2}$, $c_3 = 1$ and $L_1(v) = v_1$,
    $L_2(v) = v_1+v_2$, $L_3(v) = v_1+v_2+v_3$, so that $A$ is the lower
triangular matrix all of whose entries on and below the diagonal are
$1$. The columns of $A^{-1}$ then give $e_1 = b_1-b_2$,
$e_2 = b_2-b_3$, $e_3 = b_3$, an $N$-orthogonal basis.
\end{example}

\begin{prop}
    \label{prop:stability}
    If $(e_1,\dots,e_n)$ is an $N$-orthogonal basis of $V$ and
    $e_i' = e_i + \varepsilon_i$ with $N(\varepsilon_i) < N(e_i)$, then $(e_i')$ is $N$-orthogonal.
\end{prop}
\begin{proof}
$N(\sum \alpha_i \varepsilon_i) \leq \max_i \vert \alpha_i\vert_p N(\varepsilon_i)
< \max_i \vert \alpha_i \vert_p N(e_i)
= N(\sum \alpha_i e_i)$, and the
ultrametric equality concludes.
\end{proof}

This is very good news: unlike in the real case, no infinite precision
is needed to know orthogonality exactly. We finally record three lemmas
that will be used repeatedly.

\begin{lem}[\cite{deng2025p,zhang2024explicit}]
\label{lem:unit-criterion}
\label{lem:same-norm}
Let $N$ be a norm on $V$. (i) A basis $(b_1,\dots,b_n)$ of $V$ is
$N$-orthogonal if and only if
$N\big( \sum_{i=1}^{n} \alpha_i b_i \big) = \max_{1 \leq i \leq n} N(\alpha_i b_i)$
for all $\alpha_1,\dots,\alpha_n$ such that one of them is a unit and
the other $n-1$ lie in $\Z_p$. (ii) A family $(b_1,\dots,b_m)$ with
$N(b_1) = \cdots = N(b_m) =: \rho$ is $N$-orthogonal if and only if
$p \mid \alpha_i$ for $1 \leq i \leq m$ whenever
$\alpha_1,\dots,\alpha_m \in \Z_p$ satisfy
$N\big( \sum_{i=1}^{m} \alpha_i b_i \big) < \rho$.
\end{lem}

\begin{lem}
\label{lem:finite-test}
Let $(b_1,\dots,b_m)$ be a family of vectors of $V$ with
$N(b_1) = \cdots = N(b_m) = \rho$. Then $(b_1,\dots,b_m)$ is
$N$-orthogonal if and only if
$N\Big(\sum_{i=1}^m \bar{\alpha}_i b_i\Big) = \rho$
for all $\bar{\alpha}_i \in \{0,\dots,p-1\}$, not all zero.
\end{lem}
\begin{proof}
By Lemma~\ref{lem:same-norm}, the only combinations that can make the
norm drop are those whose scalars are all divisible by $p$; so if there
are $\alpha_i \in \{0,\dots,p-1\}$, not all zero, with
$N(\sum_i \alpha_i b_i) < \rho$, the family is not $N$-orthogonal.
Conversely, if it is not $N$-orthogonal, there are
$a_i = \alpha_i + p\mu_i \in \Z_p$, with $\mu_i \in \Z_p$ and at least
one $\alpha_i \in \{0,\dots,p-1\}$ nonzero, such that
$N(\sum_i a_i b_i) < \rho$; since $N(p\mu_i b_i) \leq \rho/p < \rho$,
the strict ultrametric inequality gives
$N(\sum_i \alpha_i b_i) < \rho$.
\end{proof}

\subsection{Commensurability and orthogonality}

Let $N$ be a norm on $V$ and let $(A,c)$ be a diagonal representation
of $N$. We write $\kappa_i = c_i p^{\Z} \in \R_{>0}/p^{\Z}$ for the
classes induced by the values $c_i$, and call them
\emph{commensurability classes}. The set
$\Gamma := \bigcup_i c_i p^{\Z} \cup \{0\}$ of values taken by $N$
being discrete, the balls
$B_{\leq \rho} := \{v \in V : N(v) \leq \rho\}$, which are free
$\Z_p$-modules, stagnate and then grow by sudden jumps as the radius
$\rho$ increases: one gets an infinite chain of strictly nested
lattices, that is, a filtration, whose associated graded piece
$\gr_{\rho} := B_{\leq \rho} / B_{< \rho}$ is an $\F_p$-vector space of
dimension
\[ \nu_{\rho} := \card\{i : c_i \in \rho\, p^{\Z}\}, \]
the multiplicity of the class of $\rho$. We write in particular
$\nu_{\max} := \max_{\rho} \nu_{\rho} \in \{1,\dots,n\}$; for
instance the parameters $c_1 = \pi$, $c_2 = \sqrt 2$, $c_3 = 1$ of
Example~\ref{ex:norm} give three distinct classes modulo $p^{\Z}$, so
that $\nu_{\max} = 1$ there.

We are also interested in the commensurability of given vectors of $V$:
two nonzero vectors $v_1,v_2 \in V$ are \emph{commensurable} (with
respect to $N$) when the ratio of their norms lies in the image of
$\vert \cdot \vert_p$, that is when $\frac{N(v_1)}{N(v_2)} \in p^{\Z}$,
and \emph{incommensurable} otherwise. This is clearly an equivalence
relation, and here is why it is relevant to orthogonality.

\begin{prop}
\label{prop:incomm-ortho}
Let $N$ be a norm on $V$. Two incommensurable vectors
$v_1,v_2 \in V$ are $N$-orthogonal.
\end{prop}
\begin{proof}
Let $\alpha_1,\alpha_2 \in \Q_p$; if one of them vanishes, the equality
$N(\alpha_1 v_1 + \alpha_2 v_2) = \max(N(\alpha_1 v_1), N(\alpha_2 v_2))$
is obvious. Otherwise $\vert \alpha_1/\alpha_2\vert_p \in p^{\Z}$ while
$N(v_1)/N(v_2) \notin p^{\Z}$, so that
$\frac{N(\alpha_1 v_1)}{N(\alpha_2 v_2)} \notin p^{\Z}$; in particular
this ratio cannot equal $1$, the two norms differ, and the ultrametric
equality gives the claim.
\end{proof}

\begin{rem}
\label{rem:commensurability}
(i) The converse is false: for the usual sup norm
$\Vert v \Vert_{\infty} = \max_{1 \leq i \leq n} \vert v_i \vert_p$ on
$\Q_p^n$ all nonzero vectors are commensurable, yet many of them are
$N$-orthogonal.
(ii) A norm oracle is enough to determine commensurability.
(iii) If the vectors of a family $(b_1,\dots,b_m)$ are pairwise
incommensurable, then the family is $N$-orthogonal: this is the case
$k = m$ of Lemma~\ref{lem:blocks} below.
\end{rem}

Commensurability classes split orthogonality into independent blocks, which is the key idea behind all the algorithms of this
paper.

\begin{lem}[block decomposition]
\label{lem:blocks}
Let $b_1,\dots,b_m$ be nonzero vectors of $V$ and let
$B_1,\dots,B_k$ be their commensurability classes. Then
$(b_1,\dots,b_m)$ is $N$-orthogonal if and only if each block $B_i$ is
$N$-orthogonal.
\end{lem}
\begin{proof}
A subfamily of an $N$-orthogonal family is $N$-orthogonal, which gives
one implication. Conversely, assume that every block is
$N$-orthogonal and let $b = \sum_{j=1}^m \alpha_j b_j$. Grouping the
terms blockwise, write $b = \sum_{i=1}^k w_i$ with
$w_i = \sum_{j \in B_i} \alpha_j b_j$. By orthogonality of the block
$B_i$ we have $N(w_i) = \max_{j \in B_i} N(\alpha_j b_j)$, which lies
in the class $\kappa_i$ of $B_i$ whenever $w_i \neq 0$. The classes
$\kappa_1,\dots,\kappa_k$ being pairwise distinct, the nonzero
$N(w_i)$ are pairwise distinct, so that
$N(b) = \max_{1 \leq i \leq k} N(w_i) = \max_{1 \leq j \leq m} N(\alpha_j b_j)$.
\end{proof}

Given a family $B$ of vectors of $V$, we write $B_1,\dots,B_k$
($k \leq m$) for the commensurability classes of its vectors, $l_i$ for
the number of vectors in $B_i$, and
$l_{\max} := \max_i \card B_i$. These classes can be
computed from the oracle alone, by Algorithm~\ref{algo:classes} above; they span subspaces which are automatically $N$-orthogonal to one another, and by Lemma~\ref{lem:blocks} nothing else has to be checked outside the blocks. Inside a block, however, nothing is known yet: one has to refine the criterion and look for a cancellation of the leading
terms in a linear combination. Methods
for this are given in Section~\ref{sec:decide}.

\begin{algorithm}[t]
\caption{Commensurability classes}
\label{algo:classes}
\KwIn{$m$ vectors $b_1,\dots,b_m$ of $V$, a norm oracle for $N$}
\KwOut{the list of the commensurability classes}
precompute the norms $N(b_j)$, $1 \leq j \leq m$\;
$\mathit{Classes} \gets [\,]$; $\mathit{rest} \gets (b_1,\dots,b_m)$\;
\While{$\mathit{rest} \neq [\,]$}{
  pick $b_0 \in \mathit{rest}$ and let $C$ be the set of the
  $b \in \mathit{rest}$ such that
  $N(b)/N(b_0) = p^{k_b} \in p^{\Z}$\;
  append $\{p^{k_b}b : b \in C\}$ to $\mathit{Classes}$
  \tcp*{all of norm $N(b_0)$}
  $\mathit{rest} \gets \mathit{rest} \setminus C$\;
}
\Return $\mathit{Classes}$\;
\end{algorithm}

Algorithm \ref{algo:classes} computes the commensurability classes
of a family $B$ of $m$ vectors in $m$ oracle queries and
at most $\binom{m}{2} = O(m^2)$ divisions.
In the most expensive case, all the vectors of $B$ are pairwise incommensurable,
and thus, one gets that $B$ is orthogonal for free.

\begin{example}
\label{ex:classes}
Take again the norm
$N(v) = \max(\pi \vert v_1\vert_p, \sqrt{2}\,\vert v_1+v_2\vert_p, \vert v_1+v_2+v_3\vert_p)$
of Example~\ref{ex:norm} and consider the basis
\[ B = \begin{pmatrix}
    1 & 1+p^2 & p^{-2} \\
    p^4 & p^3 & 1 \\
    p^{-2}+p & p+p^2 & 1
\end{pmatrix}, \]
whose vectors are the columns. One has
$N(b_1) = p^2$, $N(b_2) = \pi$ and $N(b_3) = \pi p^2$.
As such, the commensurability classes are $ \left\lbrace b_1 \right\rbrace$
and $\left\lbrace b_2, b_3 \right\rbrace$
and $l_{\max} = 2$.
One can normalize $b_3$ by replacing it
with $b_3' = p^2 b_3 = (1,p^2,p^2)^T$.
We will see in Example~\ref{ex:residual} that the basis
$(b_1,b_2,b_3')$ is \emph{not} $N$-orthogonal.
\end{example}

\subsection{\ensuremath{p}-adic lattices, LVP and CVP}

We first recall the definition of a $p$-adic lattice.

\begin{defn}[$p$-adic lattice,~\cite{deng2018some}]
Let $N$ be a norm on $V$ and let $b_1,\dots,b_m$ be $m$ linearly
independent vectors of $V$. A $p$-adic lattice in $V$ is the set
\[ \mathcal{L}(b_1,\dots,b_m)
:= \Big\{ \sum_{i=1}^{m} \alpha_i b_i : \alpha_i \in \Z_p \Big\}
= \sum_{i=1}^{m} \Z_p b_i \]
of all $\Z_p$-linear combinations of $b_1,\dots,b_m$. The integers $m$
and $n$ are called respectively the rank and the dimension of the
lattice. When $n = m$, the lattice is said to be of full rank.
\end{defn}

Every $p$-adic lattice is a compact subset of $V$.
Moreover, by~\cite[p.~72]{robert2000course}, if $\Omega \subset V$ is
compact then $\{N(x-a) : x \in \Omega\}$ is finite for every
$a \in V \setminus \Omega$, and $\{N(x-a) : x \in \Omega \setminus \{a\}\}$
is discrete in $\R_{>0}$ for every $a \in \Omega$; this is what makes
the two problems below well defined. Finally, if $(b_1,\dots,b_m)$ is
an $N$-orthogonal basis of the subspace
$\Vect_{\Q_p}(b_1,\dots,b_m) = \sum_{i=1}^{m}\Q_p b_i$ of $V$, it is
called an $N$-orthogonal basis of the lattice
$\mathcal{L}(b_1,\dots,b_m)$~\cite{deng2025p}.

\begin{defn}[LVP and CVP,~\cite{zhang2025norm}]
Let $\mathcal{L}=\mathcal{L}(b_1,\dots,b_m)$ be a $p$-adic lattice
in $V$. Set $\lambda_1=\max_{1\leq i\leq m}N(b_i)$ and
$ \lambda_2=\max\{N(v):v\in\mathcal{L},\ N(v)<\lambda_1\}.$
 The \emph{longest vector problem} (LVP) asks for a vector
  $v\in\mathcal{L}$ with $N(v)=\lambda_2$, the second-largest
  distinct norm value attained on $\mathcal{L}$.
 Given a target vector $t\in V$, the \emph{closest vector
  problem} (CVP) asks for a vector $v\in\mathcal{L}$ satisfying $N(t-v)=\min_{w\in\mathcal{L}}N(t-w).$
\end{defn}

Deng et al.~\cite{deng2018some} gave deterministic exponential
algorithms for LVP and CVP. These problems are therefore natural candidates for cryptographic hardness assumptions. Since both become polynomial-time solvable when an $N$-orthogonal basis is available~\cite{zhang2025norm}, understanding the complexity of computing such a basis is essential for assessing the hardness of $p$-adic lattices constructions. We thus focus on the problem of computing an $N$-orthogonal basis of $V$ and on its associated computational complexity.

\section{A family of norms: \texorpdfstring{$\NPV$}{N\_P}}
\label{sec:family}

To provide lower complexity bounds for the proofs of the subsequent sections,
we define a family of norms whose elements are exponentially hard to
distinguish. It is indexed by the vector lines of the projective space
$P(\F_p^n)$.

\begin{defn}[the family $\NPV$]
For a vector line $x \in P(\F_p^n)$, take $\overline{v} \in x$,
$\overline{v} \neq 0$, and lift it to an element
$v = \phi(x) \in \Z_p^n$. By construction $v$ is \emph{primitive} ---
its coordinates lie in $\Z_p$ and at least one of them belongs to
$\Z_p^{\times}$ --- so that it can be completed into a basis
$B_x = (v, v_{2,x},\dots,v_{n,x})$ of $\Z_p^n$. We define the norm
$N_x$ on $\Q_p^n$ as follows: if
$z = \alpha_1 v + \alpha_2 v_{2,x} + \dots + \alpha_n v_{n,x}$, then
\[ N_x (z) := \max \left( \vert \alpha_1 \vert_p,\ p\vert \alpha_2 \vert_p,\ \dots,\ p\vert \alpha_n \vert_p \right), \]
and we set $\NPV := \left\lbrace N_x : x \in P(\F_p^n) \right\rbrace$.
\end{defn}

Each $N_x$ is indeed an ultrametric norm on $\Q_p^n$: with the notation
of Theorem~\ref{thm:diagonal}, it is given in the basis $B_x$ by
$c = (1,p,\dots,p)$ and $a_{r,i} = \delta_{r,i}$. In particular all the
norms of $\NPV$ take their values in $p^{\Z}$, and they all satisfy
$\nu_{\max} = n$: the family is as far as possible from the
incommensurable (easy) case.

\begin{lem}
\label{lem:computing-Nx}
Let $x \in P(\F_p^n)$ and $v = \phi(x)$. The unit ball of $N_x$ is
$\Lambda_x = \Z_p v + p \Z_p^n$. Consequently, for a primitive vector
$z \in \Z_p^n$,
\begin{equation}
\label{eq:Nx-cases}
N_{x}(z) = \begin{dcases}
        1 & \textrm{if } \overline{z} \in x, \\
        p & \textrm{otherwise},
    \end{dcases}
\end{equation}
and the norms of $\NPV$ are pairwise distinct. In particular
$\card(\NPV) = \frac{p^n-1}{p-1}$.
\end{lem}
\begin{proof}
Let $z = \alpha_1 v + \alpha_2 v_{2,x} + \dots + \alpha_n v_{n,x}$ be written in the
basis $B_x$. Then $N_x(z) \leq 1$ if and only if
$\vert \alpha_1 \vert_p \leq 1$ and $p \vert \alpha_i\vert_p \leq 1$ for
$i \geq 2$, that is, $\alpha_1 \in \Z_p$ and $\alpha_i \in p\Z_p$ for $i \geq 2$;
whence the unit ball.

Let now $z \in \Z_p^n$ be primitive, so that all the $\alpha_i$ lie in
$\Z_p$. Since $(\overline v, \overline{v_{2,x}},\dots,\overline{v_{n,x}})$
is a basis of $\F_p^n$, we have $\overline{z} \in x$ if and only if
$\overline{\alpha_1} \neq 0$ and $\overline{\alpha_2} = \dots = \overline{\alpha_n} = 0$.
In that case $\vert \alpha_1\vert_p = 1$ and $p\vert \alpha_i \vert_p \leq 1$ for
$i \geq 2$, so $N_x(z) = 1$; otherwise there is some $j \geq 2$ with
$\overline{\alpha_j} \neq 0$, so that $p \vert \alpha_j\vert_p = p$, while
$\vert \alpha_1\vert_p \leq 1$ and $p\vert \alpha_i\vert_p \leq p$ for every $i$,
whence $N_x(z) = p$. Finally, for $x \neq y$ in $P(\F_p^n)$, the
primitive vector $v = \phi(x)$ satisfies $N_x(v) = 1$ (as
$\overline v \in x$) and $N_y(v) = p$ (as $\overline v \notin y$);
hence $N_x \neq N_y$, and the cardinality is that of $P(\F_p^n)$.
\end{proof}

\begin{prop}
\label{prop:indistinguishable}
Set $m = \frac{p^n-1}{p-1}$. For any $s \leq m-2$ nonzero vectors
$w_1,\dots,w_s$ of $\Q_p^n$, there exist two distinct norms of $\NPV$
which coincide on each of them. Consequently, at least
$m - 1 = p + p^2 + \dots + p^{n-1}$ vectors are needed to tell apart all
the norms of $\NPV$, and this bound is sharp.
\end{prop}
\begin{proof}
For $i \in \llbracket 1,s\rrbracket$, write $w_i = p^{k_i}w_{i,0}$ with
$w_{i,0} \in \Z_p^n$ primitive, and let $x_i \in P(\F_p^n)$ be the
vector line defined by $\overline{w_{i,0}}$. By
Lemma~\ref{lem:computing-Nx} there are $m$ distinct norms in $\NPV$,
and the set $\{x_1,\dots,x_s\}$ has at most $m-2$ elements, so there
are two distinct lines
$y_1 \neq y_2 \in P(\F_p^n) \setminus \{x_1,\dots,x_s\}$. For every
$i$, $\overline{w_{i,0}}$ belongs neither to $y_1$ nor to $y_2$, so
that $N_{y_1}(w_{i,0}) = N_{y_2}(w_{i,0}) = p$
by~\eqref{eq:Nx-cases}; by homogeneity
$N_{y_1}(w_i) = p^{-k_i}N_{y_1}(w_{i,0}) = p^{-k_i}N_{y_2}(w_{i,0}) = N_{y_2}(w_i)$,
and both norms coincide on the $w_i$'s. Conversely, $m-1$ vectors
attached to $m-1$ pairwise distinct projective lines do separate the
family: given $y_1 \neq y_2$, at least one of these two lines is
represented, say $y_1 = x_i$, and then
$N_{y_1}(w_{i,0}) = 1 \neq p = N_{y_2}(w_{i,0})$.
\end{proof}

\section{Orthogonal basis computation}
\label{sec:ortho}

Recall that knowing an $(A,c)$-representation of the norm $N$ is
equivalent to knowing an $N$-orthogonal basis, by
Theorem~\ref{thm:diagonal} and its proof; in that model,
orthogonalisation --- hence LVP and CVP, by~\cite{zhang2025norm} ---
is polynomial. We therefore assume in this section that only a norm
oracle is available.

\subsection{Lower bound: no polynomial algorithm}

\begin{thm}[Lower Bound Orthogonalization]
    \label{thm:lowerbound-ortho}
    Any algorithm (deterministic, or randomised with success probability
    $\geq 2/3$) which reconstructs an $(A,c)$-representation of $N$ from a
    norm oracle performs $\Omega(p^{\nu_{\max}-1})$ oracle calls in the
    worst case. In particular, when $\nu_{\max} = n$, the algorithm performs $\Omega(p^{n-1})$ oracle calls.
\end{thm}
The idea is that, on any nonzero query, all norms in $\NPV$ return the same answer, except the one whose secret direction matches the queried vector. Whenever the oracle returns the answer shared by all the other norms, the query eliminates at most one candidate. After fewer than $\frac{p^n-1}{p-1}-1$ such queries, at least two distinct norms remain consistent with all the answers received, and the algorithm cannot distinguish them.
\begin{proof}
Assume first $\nu_{\max}=n$ and set
$K=\card(\NPV)=\frac{p^n-1}{p-1}$.
By the argument of Proposition~\ref{prop:indistinguishable},
each nonzero query $w=p^k w_0$, with $w_0$ primitive, tests only
the line $\F_p\overline{w_0}$.

Let $\mathcal A$ be deterministic and make at most $T$ queries.
Consider its execution when each such query receives the answer
$p^{1-k}$. If $T<K-1$, at least $K-T$ distinct norms produce the same
execution and output, so $\mathcal A$ cannot be correct on
every input. Thus deterministic reconstruction requires
$T\geq K-1$.

Moreover, $\mathcal A$ can be correct on at most $T$ norms
whose lines were tested and at most one of the remaining norms.
Under the uniform distribution on $\NPV$, its success
probability is therefore at most $(T+1)/K$. Yao's principle
gives $T\geq 2K/3-1$ for randomised algorithms succeeding
with probability at least $2/3$ on every input.
Both bounds are $\Omega(p^{n-1})$.

For $2\leq\nu=\nu_{\max}<n$, apply the same construction
in the first $\nu$ coordinates and define
\[ N(z) := \max\big(N_x(z_1,\dots,z_\nu),\, c_{\nu+1}\vert z_{\nu+1}\vert_p, \dots, c_n \vert z_n\vert_p\big), \]
where $c_{\nu+1},\dots,c_n$ and $1$ are pairwise incommensurable.
The largest commensurability multiplicity is then $\nu$.
The last $n-\nu$ coordinate directions are orthogonal to the
first $\nu$-dimensional block and to one another. Their weights
are fixed and may be assumed known, so the only unknown part
of the norm lies in the first block. On each query, all but
at most one norm in this extended family still return the same
answer. Since the multiplicity of the class $p^{\Z}$ is
$\nu=\nu_{\max}$, the same counting argument gives
$\Omega(p^{\nu-1})$ queries.
\end{proof}

\subsection{The algorithm}
\label{sec:descend}

To lighten the notation, we write $\gamma_i = N(b_i)$ where $(b_i)$ is
a basis of $V$. Here is the fundamental criterion used by the
algorithm.

\begin{lem}
\label{lem:no-local-min}
Let $(b_i)$ be a basis of $V$ endowed with a norm $N$ such that, for
every $i$,
\[ N(b_i) = \min_{\delta \in W_i} N(b_i-\delta),\qquad
   W_i = \Vect(b_j)_{j\neq i}. \]
Then $(b_i)$ is $N$-orthogonal.
\end{lem}
\begin{proof}
Set $\gamma_i=N(b_i)$. If $\delta\in W_i$ satisfies
$N(\delta)>\gamma_i$, the ultrametric equality gives
$N(b_i-\delta)=N(\delta)>\gamma_i$.
Since $\delta=0$ gives the value $\gamma_i$, the minimum
may be sought over the compact set
$W_i\cap B_{\leq\gamma_i}$, and is therefore attained.
Assume now that there is $v = \sum \alpha_i b_i \in V$ with
$N(v) < \rho := \max_i \vert \alpha_i\vert_p \gamma_i$ (the inequality
$\leq$ being automatic), let $i_0$ achieve the maximum, so that
$\alpha_{i_0} \neq 0$, and put
$\delta = -\sum_{j \neq i_0} \frac{\alpha_j}{\alpha_{i_0}} b_j \in W_{i_0}$.
Then
$N( b_{i_0} - \delta) = \vert \alpha_{i_0}\vert_p^{-1} N(v)
< \vert \alpha_{i_0}\vert_p^{-1} \rho = \gamma_{i_0}$, contradicting
the minimality of $N(b_{i_0})$.
\end{proof}

In other words, local descent implies global optimality: no random
restart is needed, in contrast with real gradient descent --- an
important point, since the orthogonalisation algorithm with an oracle
is exactly a discrete $p$-adic descent: one starts from a defect of
minimality in the norms of the $b_i$ and corrects it by finding
$\delta \in W_i$ making $N(b_i-\delta)$ minimal. Such a correction is
exactly an instance of CVP, the ultrametric inequality restricting the
search space to the ball $W_i \cap B_{\leq N(b_i)}$, a finitely
generated $\Z_p$-module; it is this geometric reduction of a vector
problem to a lattice problem that gives a transparent intuition for the
polynomial equivalence, already known, between orthogonalisation and
CVP.

The next lemma is the engine of the algorithm: it makes the descent
test finite and confines it to a single commensurability block.
\begin{lem}
\label{lem:descent-test}
Let $(e_1,\dots,e_{k-1})$ be an $N$-orthogonal family,
$\gamma_j = N(e_j)$, and let
$x \notin W = \Vect(e_1,\dots,e_{k-1})$ with $\gamma = N(x)$. Set
\[ J = \Big\{ j < k \;:\; \frac{\gamma_j}{\gamma} \in p^{\Z} \Big\},
\qquad k_j := \log_p(\gamma_j/\gamma) \ \ (j \in J), \]
so that $N(p^{k_j}e_j) = N(x)$ for $j \in J$. Then
$(e_1,\dots,e_{k-1},x)$ is $N$-orthogonal if and only if
\[ N\Big(x - \sum_{j\in J}\beta_j\,p^{k_j}e_j\Big) = \gamma
\qquad \text{for all } \beta_j \text{'s in}  \{0,1,\dots,p-1\}. \]
Moreover $\card J \leq \nu_{\gamma} \leq \nu_{\max}$.
\end{lem}
\begin{proof}
By Lemma~\ref{lem:blocks}, $(e_1,\dots,e_{k-1},x)$ is
$N$-orthogonal if and only if each of its commensurability blocks is.
The blocks that do not contain $x$ are subfamilies of
$(e_1,\dots,e_{k-1})$, hence are orthogonal; the block containing $x$
is $\{x\} \cup \{e_j : j \in J\}$, which after the normalisation
$e_j \mapsto p^{k_j}e_j$ --- a rescaling that changes neither the block
nor its orthogonality --- consists of vectors of norm $\gamma$. By
Lemma~\ref{lem:finite-test}, that block is orthogonal if and only if
$N(\bar\alpha_0 x + \sum_{j \in J}\bar\alpha_j p^{k_j}e_j) = \gamma$
for all $\bar\alpha \in \{0,\dots,p-1\}^{J \cup \{0\}}$ not all zero;
if $\bar\alpha_0 = 0$ this holds automatically, the $e_j$ being already
orthogonal, and if $\bar\alpha_0 \neq 0$, dividing by the unit
$\bar\alpha_0$ brings us back to the stated form.
Finally, the same criterion applied to the family $(p^{k_j} e_j)_{j \in J}$ shows that every nonzero $\F_p$-linear combination of these vectors has norm $\gamma$ by orthogonality. Therefore, their images in $\text{gr}_{\gamma} = B_{\leq \gamma} / B_{< \gamma}$ are $\F_p$-independent,
thus $\card J \leq \nu_{\gamma}$.
\end{proof}

We can now build in Algorithm \ref{algo:descend} a recursive algorithm
to compute an $N$-orthogonal basis
$(e_1,\dots,e_n)$ out of a basis $(b_1,\dots,b_n)$ of $V$. At step $k$,
assuming $(e_1,\dots,e_{k-1})$ orthogonal, it checks whether the norm
of $x := b_k$ can be made to drop by subtracting an $\F_p$-linear
combination of those $e_j$ that are commensurable with it; if so it
replaces $x$ by the new vector (it goes down one ``floor'') and starts
again, until no further descent is possible, i.e.\ until the global
minimum has been found. When no $\F_p$-linear combination makes the
norm drop, $e_k := x$ is orthogonal to $(e_1,\dots,e_{k-1})$ by
Lemma~\ref{lem:descent-test}.

\begin{algorithm}[t]
\caption{\textsc{Descend}: orthogonalisation with an oracle}
\label{algo:descend}
\KwIn{a norm oracle for $N$ on $V$, a basis $B = (b_1,\dots,b_n)$}
\KwOut{an $N$-orthogonal basis $(e_1,\dots,e_n)$ of $V$}
$E \gets \emptyset$; $W \gets \{0\}$\;
\For{$k = 1,\dots,n$}{
  $x \gets \argmin_{b_j \notin W} N(b_j)$; $\gamma \gets N(x)$\;
  \tcp{descent to a best approximation (CVP)}
  \While{\textnormal{true}}{
    $J \gets \{ j < k : \gamma_j/\gamma \in p^{\Z}\}$\;
    $k_j \gets \log_p(\gamma_j/\gamma)$ for $j \in J$\;
    $\mathit{found} \gets \textnormal{false}$\;
    \ForEach{$\beta \in \{0,\dots,p-1\}^J \setminus \{0\}$}{
      $y \gets x - \sum_{j \in J} \beta_{j}\, p^{k_j} e_j$\;
      \If{$N(y) < \gamma$}{
        $x \gets y$; $\gamma \gets N(y)$;
        $\mathit{found} \gets \textnormal{true}$; \textbf{break}\;
      }
    }
    \lIf{\textnormal{not} $\mathit{found}$}{\textbf{break}
      \tcp*[f]{$x$ is optimal}}
  }
  $e_k \gets x$; $\gamma_k \gets \gamma$;
  $W \gets W \oplus \Q_p \cdot e_k$\;
}
\Return $E = (e_1,\dots,e_n)$\;
\end{algorithm}

To evaluate the complexity of this algorithm, we must bound the number
of times the norm of a vector can drop.
At the $k$-th iteration, let $x_k^{(0)}$ denote the vector selected at step 3 of the algorithm, $\gamma$ its norm and let $W_{k-1}$ be the subspace
spanned by the first $k-1$ orthogonalised vectors.
Since $x_k^{(0)} \notin W_{k-1}$ and $W_{k-1}$ is closed, the distance between $x_k^{(0)}$ and $W_{k-1}$, $d_k := \min_{w \in W_{k-1}} N(x_k^{(0)} - w)$
is strictly positive.
Recall that the set of values taken by $N$ is contained in
$\Gamma = \bigcup_{r=1}^n c_r p^{\Z} \cup \{0\}$. For each $r$, the class $c_r p^{\Z}$ contains at most one element in the interval $\,] \gamma / p, \gamma]$,
since two consecutive elements of this class differ by a factor $p$.
Hence the norm can take at most $n$ distinct values on an interval
$\,]\gamma/p, \gamma]$. Since each update $x \leftarrow x - w$ makes the norm drop strictly, after at most $n$ updates the norm is at most $\gamma /p$.
Repeating this argument, after $t$ updates the working vector $x_k^{(t)}$ satisfies $N(x_k^{(t)}) \leq N(x_k^{(0)}) p^{- \lfloor t/n \rfloor}$. On the other hand, every working vector is of the form $x_k^{(0)} - w$, $w \in W_{k-1}$, so its norm is at least $d_k$. Therefore the number of updates is at most $n \lceil \log_p ( N(x_k^{(0)})/d_k ) \rceil$.

We thus define the
\emph{geometric conditioning} of the instance as the maximal number of
descents,
\[ D = \max_{2 \leq k \leq n} n \left\lceil \log_p \left( \frac{N(x_k^{(0)})}{d_k} \right) \right\rceil, \]
a parameter which is finite and fixed by the input.

\begin{prop}[termination, complexity, correctness]
\label{prop:descend}
Algorithm~\ref{algo:descend} is correct and  terminates after finitely many steps, with
a query complexity of
\[ \#\text{oracle calls} = O\big(n(1 + D) \cdot p^{\nu_{\max}}\big). \]
\end{prop}
\begin{proof}
By definition of $D$, for each
vector $k \leq n$ the algorithm finds at most $D$ successive combinations that
decrease the norm, hence performs at most $1+D$ complete explorations
of a level (the last one certifying optimality). During one
exploration, the space of coefficients $\F_p^J$ has size at most
$p^{\nu_{\max}}$ by Lemma~\ref{lem:descent-test}, so a level costs
$O(p^{\nu_{\max}})$ oracle calls; summing over the $n$ basis vectors
gives the announced complexity, and in particular termination.

Finally, assume by induction that $(e_1,\dots,e_{k-1})$ is
orthogonal. At step $k$ the loop stops when the search over $J$ fails, which by Lemma~\ref{lem:descent-test} means exactly that
$(e_1,\dots,e_{k-1},e_k)$ is $N$-orthogonal.
\end{proof}

Comparing with the lower bound $\Omega(p^{\nu_{\max}-1})$, the
algorithm is optimal up to a factor $\mathrm{poly}(n,D)\cdot p$. Certifying that a basis is $N$-orthogonal, on the other
hand, requires solving the orthogonality decision problem, which is the subject of the next section.

\begin{example}
\label{ex:descend}
Take again the norm
$N(v) = \max(\pi \vert v_1\vert_p, \sqrt{2}\vert v_1+v_2\vert_p, \vert v_1+v_2+v_3\vert_p)$,
assuming $p \geq 5$, and the basis $B$ of Example~\ref{ex:classes}.
Since $b_1$ and $b_2$ belong to distinct commensurability classes
(generated by $1$ and $\pi$), they are $N$-orthogonal from the start
and we set $e_1 = b_1$, $e_2 = b_2$. The vector $b_3$, on the other
hand, is commensurable with $e_2$: Algorithm~\ref{algo:descend} first
normalises it as $b_3' = p^2 b_3 = (1,p^2,p^2)^T$, of norm $\pi$, then
kills the dependency in the residual space of the class $\pi$ by a
cancellation modulo $p$,
\[ x = b_3' - e_2 = (-p^2,\ p^2-p^3,\ -p)^T, \]
whose norm drops strictly to
$N(x) = \max(\pi p^{-2}, \sqrt{2}p^{-3}, p^{-1}) = p^{-1}$ (because
$p \geq 5$): the vector $x$ has just changed commensurability class and
joined that of $e_1$. A second elimination, now in the class of $1$,
gives
\[ e_3 = x + p^3 e_1 + p^2 e_2
= (p^3+p^4,\ p^2-p^3+p^5+p^7,\ p^3+2p^4)^T, \]
of norm $\max(\pi p^{-3}, \sqrt{2}p^{-2}, p^{-2}) = \sqrt{2}p^{-2}$:
$e_3$ has reached the class of $\sqrt 2$, which was vacant,
so that $(e_1,e_2,e_3)$ is an $N$-orthogonal basis by
Remark~\ref{rem:commensurability}(iii). For readability we used $-1$ to represent $p-1$ modulo $p$; by the ultrametric inequality, it does not affect whether the norm of the corresponding
combination drops or not.
\end{example}

Finally, it is clear that $D$ depends on the chosen norm: it is the
algorithmic conditioning of the instance, and depending on $N$ it can
be arbitrarily large. This exhibits the two distinct
parameters that govern the cost. The multiplicity $\nu_{\max}$ of the
largest commensurability class yields an intrinsic barrier
$\Omega(p^{\nu_{\max}-1})$ that no algorithm can cross, whereas the
geometric conditioning $D$ is exactly what makes a uniform bound fail, as in~\cite{duan2026norm}.
From a cryptographic perspective, $p^{\nu_{\max}}$ being unconditional on
the algorithm makes it desirable as a plausible security parameter
while the handling of $D$ may need more investigation.

\section{Deciding orthogonality}
\label{sec:decide}

\subsection{The residual commensurability criterion}

In this subsection, suppose that an $(A, c)$-representation of the norm is known.
Let $B = (b_1,\dots,b_m)$ be a family of nonzero vectors of $V$ and let
$B_1,\dots,B_k$ ($1 \leq k \leq m$) be its commensurability classes. In
each block $B_i$ we may normalise every vector, $b_j' = \alpha_j b_j$,
so that all the $b_j'$ have the same norm $\rho_i > 0$; we define
$\bar M_i$, the residual matrix with coefficients in $\F_p$
containing the leading digits of the dominant $L_r(b_j')$,
as follows.

The columns of $\bar M_i$ are indexed by the indices $j$ such
that $b_j \in B_i.$
The rows of $\bar M_i$ are indexed by the indices
$r$ such that $c_r\vert L_r(b_j)\vert_p = \rho_i$
for some $b_j \in B_i.$

For such a couple of $r$ and $j,$
$\vert L_r(b_j)\vert_p = \frac{\rho_i}{c_r}.$
As $\vert \cdot \vert_p$ takes its
values in $p^{\Z} \cup \{0\}$, this forces $\rho_i/c_r \in p^{\Z}$, say
$\rho_i/c_r = p^{-d_{r}^{(i)}}$ with $d_{r}^{(i)} \in \Z$.
In that case, one can write
$L_r(b_j) = a_0 p^{d_{r}^{(i)}} + \cdots$
and $a_0 = L_r(b_j)\,p^{-d_{r}^{(i)}} \bmod p \in \F_p$.
For the same $r$ and other $j$'s of the same
commensurability class,
$\vert L_r(b_j)\vert_p \leq \frac{\rho_i}{c_r}=p^{-d_{r}^{(i)}},$
thus, $L_r(b_j)\,p^{-d_{r}^{(i)}}\bmod p \in \F_p$
and is equal to $0$ when $c_r \vert L_r(b_j)\vert_p < \rho_i.$

As such,  we can define
$\bar M_i,$ a matrix with coefficients in $\F_p$, as follows: for $j$ such that $b_j \in B_i$
and for $r$ such that $c_r \vert L_r(b_\cdot)\vert_p$ reaches $\rho_i,$
\[\left( \bar M_i \right)_{r,j}:=  L_r(b_j)\,p^{-d_{r}^{(i)}} \bmod p \in \F_p.\]

\begin{thm}[Residual Commensurability Criterion]
\label{thm:residual}
Let $N$ be a norm on $V$ and let $b_1,\dots,b_m$ be nonzero vectors of
$V$. Let $B_1,\dots,B_k$ be their commensurability classes and
$\bar M_1,\dots,\bar M_k$ the associated residual matrices of leading
digits. Then the family is $N$-orthogonal if and only if
$\ker \bar M_i = \{0\}$ over $\F_p$ for every block $B_i$.
\end{thm}
\begin{proof}
By Lemma~\ref{lem:blocks} the family is $N$-orthogonal if and only if
each block is, so we may work inside a single block of commensurability, of size $l$ say;
and the normalisation $b_j' := \alpha_j b_j$ with
$\alpha_j \in \Q_p^*$ does not affect orthogonality --- the factor is
absorbed by the coefficients of the linear combination --- so we may
assume that all the $b_j$ have norm $\rho$. By
Lemma~\ref{lem:finite-test}, the block is then $N$-orthogonal if and
only if $N(b) = \rho$ for every $b = \sum_{j=1}^l \bar\alpha_j b_j$
with $\bar\alpha \in \{0,\dots,p-1\}^l \setminus \{0\}$.

Let us assume that our block is not $N$-orthogonal.
Let $\bar\alpha \in \{0,\dots,p-1\}^l \setminus \{0\}$ be a corresponding
combination with
$N(\sum_{j=1}^l \bar\alpha_j b_j ) < \rho.$
Then $\max_r c_r \vert L_r(\sum_{j=1}^l \bar\alpha_j b_j  ) \vert_p < \rho. $
As such, for any $r,$
$\val \left(  L_r(\sum_{j=1}^l \bar\alpha_j b_j  ) \right) > - \log_p \left( \frac{\rho}{c_r} \right)=d_r.$
Consequently, for a given $r$,
$\val \left( \sum_{j=1}^l p^{-d_r}L_r(  b_j  ) \bar\alpha_j \right) >0$
and for every $j$, $\val \left( p^{-d_r}L_r(  b_j  ) \bar\alpha_j  \right) \geq 0.$
By the ultrametric inequality, it means that in $\F_p,$
$\sum_{j=1}^l  \left( p^{-d_r} L_r(  b_j) \bar\alpha_j  \bmod p \right)  = 0$
and thus
\[ \sum_{j=1}^l \bar M_{r,j} \bar\alpha_j =0.\]
This means that $\bar \alpha$ is a non-zero element of $\ker \bar M,$
which proves the first implication.

Conversely, let $\bar \alpha$ be a non-zero element of $\ker \bar M$
and let $r$ be the index of a row of $\bar M.$
We get that $\sum_{j=1}^l \bar M_{r,j} \bar\alpha_j =0.$
As before, for every $j$,
$\val \left( p^{-d_r}L_r(  b_j  ) \bar\alpha_j  \right) \geq 0$
and in addition,
$\sum_{j=1}^l  \left( p^{-d_r} L_r(  b_j) \bar\alpha_j  \bmod p \right) = 0$.
This means that
$\left( \sum_{j=1}^l  \left( p^{-d_r} L_r(  b_j) \bar\alpha_j   \right) \right) \bmod p = 0$,
and thus $\val \left( \sum_{j=1}^l p^{-d_r}L_r(  b_j  ) \bar\alpha_j \right) >0.$
Consequently, for any $r,$
$\val \left(  L_r(\sum_{j=1}^l \bar\alpha_j b_j  ) \right) > - \log_p \left( \frac{\rho}{c_r} \right)=d_r$,
and finally,
$\max_r c_r \vert L_r(\sum_{j=1}^l \bar\alpha_j b_j  ) \vert_p < \rho, $
which concludes.
\end{proof}

Theorem~\ref{thm:residual} gives an algorithm to test orthogonality:
compute the commensurability classes, build the matrix $\bar M_i$ of
each block from the formula above --- reading off the valuations and
the leading digits of the $L_r(b_j)$, which costs $O(nl)$ basic
operations for a block of size $l$ --- and compute its kernel.

One optimisation can be added: a matrix
with strictly more columns than rows already has a non-trivial kernel.
As such, kernel computation can be avoided for the so-called `wide' matrices.
The resulting Algorithm is presented as Algorithm \ref{algo:decide-(A,c)}.

\begin{algorithm}[t]
\caption{Residual commensurability with $(A, c)$}
\label{algo:decide-(A,c)}
\KwIn{a family $B$ of $m \geq 2$ vectors $b_j$ of size $n$, a diagonal representation $(A,c)$ of a norm $N$}
\KwOut{true if $B$ is $N$-orthogonal, false otherwise}
compute the $k \leq m$ commensurability classes with
  Algorithm~\ref{algo:classes}\;
\lIf{$\card \mathit{Classes} = m$}{\Return true
  \tcp*[f]{pairwise incommensurable}}
\ForEach{class $B_i$ of $l_i \geq 2$ vectors, of common norm $\rho_i$}{
    build $\bar M_i$, having $n_i$ rows\;
    \lIf{$n_i > l_i$ \textnormal{or} $\ker \bar M_i \neq \{0\}$}{\Return false}
}
\Return true\;
\end{algorithm}

In the worst case all the vectors are commensurable and
the kernel computation cannot be split into smaller
blocks and dominates the running time, which gives $O(nm^2)$
elementary operations. It becomes linear in the favourable cases.

\begin{example}
\label{ex:residual}
Take again the norm and the basis of Example~\ref{ex:classes}, where
$b_1$ is alone in its commensurability class while $b_2$ and $b_3$ are
commensurable with each other.
The last vector is replaced by $b_3' = p^2 b_3 = (1,p^2,p^2)^T$.
We get that $\rho_2=N(b_2)=N(b_3')=\pi.$
This value is reached only by
$\pi \vert L_1(b_2) \vert_p=\pi \vert L_1(b_3') \vert_p.$
The $c_2 \vert L_2(b_j) \vert_p $ and $c_3 \vert L_3(b_j) \vert_p $
are not even commensurable with $\pi.$
As such, $\bar M_2$
has one row and two columns, being equal to
$\bar M = \begin{pmatrix} 1 & 1\end{pmatrix}$.
 The kernel vector
$(1,-1)$ even exhibits the cancellation, $b_3' - b_2$, which
is precisely the one that made the norm drop in
Example~\ref{ex:descend}.
\end{example}

\begin{rem}
(i) For the sup norm, all vectors can be normalised so that their norm
becomes $1$. The matrix $\bar M$ is then simply the matrix of the
normalised vectors mod $p$, and the family is orthogonal if and only if
$\bar M$ has maximal rank.
(ii) This criterion requires knowing a diagonal representation of the
norm: both the $c_r$ and the $L_r$ are needed.
\end{rem}

\subsection{With a norm oracle}

The previous subsection decides orthogonality when an
$(A,c)$-representation of the norm is available. With only a norm
oracle at hand, kernel computations are out of reach.
Given a block of $l$
commensurable vectors, all of norm $\rho$ --- the commensurability
classes themselves are computed with Algorithm~\ref{algo:classes} ---
Lemma~\ref{lem:finite-test} tells us it is enough to check
whether
$N\big(\sum_{i=1}^l \alpha_i b_i\big) = \max_i \vert\alpha_i\vert_p \rho$
for every $\alpha_1,\dots,\alpha_l \in  \{0,\dots,p-1\}$, not all zero.
This naively gives
$p^l-1$ combinations.
However, if $v = \sum_{i=1}^l \alpha_i b_i$ makes
the norm drop, so does $\alpha v$ for every
$\alpha \in \{1,\dots,p-1\}$, since
$N(\alpha v) = \vert \alpha\vert_p N(v) = N(v)$; working in the
projective space avoids these $p-1$
repetitions and leaves
\[ \frac{p^l - 1}{p - 1} \]
norm computations to perform. This remains exponential in $l$, hence
not very practical. It is worth noting that it is exactly the count
obtained by Duan and Wu for the worst-case query complexity of the LVP
in a lattice of rank $l$~\cite{duan2026norm}.

We write $l_{\max}$ for the number of vectors in the largest
commensurability class; it can be any integer between $1$ and $n$,
depending entirely on the family being tested, and is known as soon as
the classes have been computed. This leads us to the following theorem.

\begin{thm}[Lower Bound Decision]
    \label{thm:lowerbound-decide}
    Any algorithm (deterministic, or randomised with success probability
    $\geq 2/3$) that decides whether a basis $B$ is $N$-orthogonal performs
    $\Omega(p^{l_{\max}-1})$ oracle calls in the worst case. In particular,
    when $l_{\max} = n$ the algorithm performs $\Omega(p^{n-1})$ oracle calls.
\end{thm}
\begin{proof}
Let us place ourselves in the worst case $l_{\max} = n$ and take for
$\mathcal{B} = (e_1,\dots,e_n)$ the canonical $\Z_p$-basis of $\Z_p^n$.
Let
$N_0(b) = p \max_{1 \leq i \leq n} \vert \alpha_i\vert_p$, for $b = \sum_{i=1}^n \alpha_i e_i.$
Then $N_0(w)=p$ for any primitive vector $w \in \Z_p^n$
and thus, thanks to Lemma \ref{lem:descent-test},
$\mathcal B$ is $N_0$-orthogonal.

Now, let us assume that there is a deterministic algorithm
that can decide orthogonality of bases of $V$
in at most $T < K_{LBD}:= \frac{p^n-1}{p-1}-n$ norm oracle queries.
We apply it to $\mathcal{B}$ and $N_0.$ It answers ``True''
and call the norm oracle for the vectors
$v_1,\dots, v_{T_0} \in V$
for some $T_0 \leq T < K$.
For any $i,$ let us write $v_i = p^{k_i} \widetilde{v_i}$ with $\widetilde{v_i}$
primitive.
Let  $x \in P(\F_p^n)$ be distinct from the residual lines $Vect( \overline{\widetilde{v_i}})$ and $Vect(\overline{e_i}) \subset \F_p^n$
and let $N_x \in \NPV$ be the corresponding norm.
Thanks to Lemma \ref{lem:computing-Nx},
$N_0(\widetilde{v_i}) = p$
and $N_x(\widetilde{v_i}) = p$,
thus
$N_0(v_i) = N_x(v_i).$
However, $B$ is not $N_x$ orthogonal
as there is a linear combination with coefficients
in $\llbracket 0, p-1 \rrbracket$
corresponding to $x$ making the norm drop.
The algorithm being deterministic and accessing norms only through
a norm oracle,
if it is given $B$ and $N_x$,
it should make the exact same operation and the same $T_0$ queries
as for $B$ and $N_0.$
As such, the exact same execution can not
provide two distinct answers (``True'' and ``False''), contradicting the hypotheses.

For the randomised case, take the distribution that picks $N_0$ with
probability $1/2$, and otherwise a uniform $N_x$ among the $K_{LBD}$
admissible ones: a deterministic algorithm making $T$ queries then errs
with probability at least $\frac{1}{2}\min\big(1, \frac{K_{LBD}-T}{K_{LBD}}\big)$,
so an error probability $< 1/3$ forces $T > K_{LBD}/3$, and Yao's principle
concludes. Finally, running the same construction inside the largest
commensurability block, of size $l_{\max}$, and completing with
pairwise incommensurable directions (which are orthogonal to everything
by Proposition~\ref{prop:incomm-ortho} and cost nothing to handle),
gives $\Omega(p^{l_{\max}-1})$.
\end{proof}

The matching algorithm is written as Algorithm~\ref{algo:decide-oracle}.

\begin{algorithm}[t]
\caption{Residual commensurability with oracle}
\label{algo:decide-oracle}
\KwIn{a family $B$ of $m \geq 2$ vectors $b_j$ of size $n$; a norm
  oracle for $N$}
\KwOut{true if $B$ is $N$-orthogonal, false otherwise}
compute the $k \leq m$ commensurability classes with
  Algorithm~\ref{algo:classes}\;
\lIf{$\card \mathit{Classes} = m$}{\Return true
  \tcp*[f]{pairwise incommensurable}}
\ForEach{class $B_i$ of $l_i \geq 2$ vectors, of common norm $\rho_i$}{
    \ForEach{combination $v$ of $B_i$ with coefficients mod $p$, up to
      projective equivalence}{
      \lIf{$N(v) < \rho_i$}{\Return false}
    }
}
\Return true\;
\end{algorithm}

Inside each block of $l_i$ commensurable
vectors of common norm $\rho_i$, it enumerates the
$\frac{p^{l_i}-1}{p-1}$ projective combinations with coefficients
modulo $p$ and returns false as soon as one of them has norm
$< \rho_i$. Its complexity is then
\[ O\left(m + \sum_{i} \frac{p^{l_i}-1}{p-1}\right). \]
In the worst case, the largest block alone contributes
$\Theta\left(\frac{p^{l_{\max}}-1}{p-1}\right) = \Theta(p^{l_{\max} - 1})$ oracle queries. This is precisely the contribution captured by our lower bound in Theorem~\ref{thm:lowerbound-decide}. When there is a single block, this gives a matching bound \[ \Theta(p^{l_{\max} - 1}). \]

\section{Conclusion and further work}
\label{sec:further}

The hardness of the orthogonalisation and decision problems comes above
all from the representation of the norm and from commensurability: when
all the vectors are pairwise incommensurable everything becomes
polynomial, whereas the complexity is exponential in $\nu_{\max}$
(resp.\ $l_{\max}$) otherwise. Moreover, the proofs of
Theorems~\ref{thm:lowerbound-ortho} and~\ref{thm:lowerbound-decide}
show that these problems amount to identifying a hidden subgroup ---
here a hidden line of $\F_p^n$ --- through membership queries, a
setting whose hardness is well understood~\cite{dutto2025survey}: it
does admit Grover-type quantum speed-ups, in $p^{(n-1)/2}$, but nothing
better is known to date and the gain remains marginal. To escape the
exponential regime one has to change the oracle model; an \emph{argmax}
oracle returning both the value of the norm and an index $r$ achieving
the maximum, for instance, would give information on the linear forms
directly, with no exhaustive search.

We therefore leave the following question open: what is the hardness of
these problems in the intermediate case where the norm is given neither
by a diagonal form (the easy case) nor by a black-box oracle (the hard
case), but by a public representation designed to hide a diagonal form?
In other words, is it possible to build a representation of a norm $N$
that reveals no $N$-orthogonal basis? This opens the way to new
asymmetric cryptographic constructions.

\printbibliography

\end{document}